\documentclass[11pt]{article}
\usepackage{graphicx}
\usepackage{amsmath}
\usepackage{dsfont}
\usepackage{amsthm}
\newtheorem{theorem}{Theorem}[section]
\usepackage[top=2.7cm,left=2.7cm, right=2.7cm, bottom=2.7cm]{geometry}

\def\beta{\lambda}
\def\gamma{\mu}

\title{ Sampling through time and phylodynamic inference with coalescent and birth-death models}
\author{Erik M Volz$^{1*}$ and Simon DW Frost$^2$\\
 { \footnotesize 1. Department of Infectious Disease Epidemiology, Imperial College London, United Kingdom } \\
 { \footnotesize 2. Department of Veterinary Medicine, University of Cambridge, Cambridge, United Kingdom }\\
 { \footnotesize   $*$ Corresponding author: e.volz@imperial.ac.uk }
}
\date{\today}

\begin{document}
\maketitle
\begin{abstract}
Many population genetic models have been developed for the purpose of inferring population size and growth rates from random samples of genetic data. 
We examine two popular approaches to this problem, the coalescent and the birth-death-sampling model, in the context of estimating population size and birth rates in a population growing exponentially according to the birth-death branching process. 
For sequences sampled at a single time, we found the coalescent and the birth-death-sampling model gave virtually indistinguishable results in terms of the growth rates and fraction of the population sampled, even when sampling from a small population. 
For sequences sampled at multiple time points, we find that the birth-death model estimators are subject to large bias if the sampling process is misspecified. 
Since birth-death-sampling models incorporate a model of the sampling process, we show how much of the statistical power of birth-death-sampling models arises from the sequence of sample times and not from the genealogical tree. 
This motivates the development of a new coalescent estimator, which is augmented with a model of the known sampling process and is potentially more precise than the coalescent that does not use sample time information.
\end{abstract}

% intro
The genetic diversity of many pathogens is influenced by recent epidemiological history, and a variety of methods exist to estimate features of an epidemic history given random samples of pathogen genetic markers \cite{volz2013viral}. An issue that is central to how pathogen genetic diversity is understood is how infected individuals are sampled. A great deal of theory has been developed under the assumption of complete sampling, that is, that all infected individuals in the population are sampled and provide at least one pathogen sequence. These methods have found great utility for the study of small outbreaks \cite{didelot2014bayesian,ypma2011unravelling}, and for certain hospital acquired infections \cite{snitkin2012tracking}. A separate body of theory has developed for the study of epidemics where a sample of hosts is obtained for pathogen sequencing, and these methods are derived from classical population genetic models such as the coalescent \cite{kingman1982coalescent,wakeley2009coalescent} and classical population dynamics models such as the birth death process \cite{rannala1996probability,kendall1948generalized}. This manuscript considers the scenario of incomplete sampling and the potentially confounding effects of non-random sampling through time on inference using the coalescent and birth-death-sampling formuli \cite{stadler2010sampling}. 

%the coalescent
The coalescent is a mathematical model of genealogies and describes the structure of genealogies generated by different demographic processes \cite{rosenberg2002genealogical}. The coalescent has been the standard tool for demographic inference and is the underlying genealogical model in most phylogenetic software \cite{guindon2010new,drummond2012bayesian}. Under the neutral coalescent, the time between consecutive common ancestry events (the internode intervals) are modeled as a point process with a hazard rate $r(t)$ that depends on the effective population size $N_e(t)$ and the number of extant lineages in that interval $A(t)$ at time $t$ in the past. With time in units of the generation interval $\tau$, this becomes
$$ r(t)  = {A(t) \choose 2} / N_e(t) $$
By relating the time of common ancestry to the population size, the coalescent enables estimation of the latter. A variety of non-parametric \cite{Pybus2000,strimmer2001exploring,opgen2005inference} and parametric \cite{donnelly1995coalescents,Pybus2000,volz2012complex} models have been developed for $N_e$ as a function of time. The parametric models for $N_e(t)$ tend to be deterministic functions of time, and we will consider such deterministic models in this paper, although there have been several recent attempts to fit stochastic demographic process models using the coalescent \cite{rasmussen2011inference,rasmussen2014phylodynamic}.

%the birth death model
Birth-death processes trace their origins to work by Kendall\cite{kendall1948generalized}, who showed how to calculate the probability that a given number of lineages will survive up to a given point of time in a stochastically growing population. Further results were developed by \cite{thompson1975human,gernhard2008conditioned}, who showed how to calculate the probability density of genealogies generated by the birth-death process process under complete sampling. These models were subsequently extended to account for incomplete sampling of the population by Stadler \textit{et al.} \cite{stadler2010sampling}. In order to account for incomplete sampling, the birth-death process must be combined with a model of the sampling process. Two sampling processes have thus far been considered in birth-death-sampling models: sampling of lineages may take place at a constant rate; or, at a given point in time, a proportion of lineages may be sampled uniformly at random. These sampling processes may be combined, and recently-developed methods allow sampling rates to vary through time according to a step function ~\cite{stadler2013birth}. 

%exactly which versions will be used, introduce bdm and com
There are many variations on the coalescent and birth death models that could be compared. 
Different coalescent and birth-death models make different assumptions about the demographic and sampling process,
 and each will be susceptible to different levels of bias by violation of those assumptions. 
We will focus on two models that have recently received considerable attention and have been used in epidemiological investigations. 
We use the birth-death-sampling model (henceforth abbreviated BDM) described in \cite{stadler2010sampling}
, and the the coalescent model (henceforth abbreviated CoM) for an unstructured population as described in \cite{volz2012complex}. 
Originally, CoMs were based on restrictive assumptions about the proportion of the population sampled and when taxa are sampled. 
CoMs were also based on strictly deterministic demographic processes,
 but all of these assumptions have been relaxed since the coalescent was first introduced. 
BDMs were originally based on census sampling at a single point in time,
 but that assumption has also been relaxed. 
Both models have been extended to consider heterogeneous structured populations \cite{volz2012complex,stadler2013uncovering}.

The likelihood of a genealogy given a demographic history may be calculated using either the CoM or the BDM, 
though these two models have very different mathematical foundations. 
The likelihood functions provided by each approach are difficult to reconcile mathematically, yet they tend to give similar results as we demonstrate below. 
The BDM has the advantage of accounting for stochasticity of the demographic process in an efficient and natural way. 
It is also possible to account for stochastically varying effective population size in the coalescent, 
but this has greater computational requirements \cite{rasmussen2011inference}. 
A potential disadvantage of BDMs is that they require a model of the sampling process, 
whereas the coalescent makes no assumptions about how lineages are sampled through time. 
If the sampling process deviates from the simplistic processes that form the basis of current BDM theory, 
it is possible that estimates based on the BDM will be biased. 

Both methods have particular advantages and vulnerabilities. 
Estimates based on CoMs may be biased by noisy demographic processes, 
and estimates based on the BDM may be biased by misspecification of the sampling process. 
In this manuscript, we will evaluate the vulnerability of both methods to these confounders. 
Because of the additional assumptions about sampling built into BDMs, 
it is difficult to make a direct comparison of the statistical power of BDMs and CoMs. 
If the sampling process is correctly specified, the observed sequence of sample times provides a great deal of information about 
the population size through time, which is not directly accessible with the CoM approach. 
Indeed, given a sequence of sample times, it is possible to estimate birth and death rates without a genealogy provided 
that the model of the sampling process is correctly specified (section~\ref{sec:tos}). 
We show that much of the statistical power of the BDM approach is derived from information in the sequence of sample times, 
and not in the genealogy. 
This finding also suggests an enhancement to CoMs: 
if the sampling process is known, we can augment the CoM likelihood with a separate likelihood for the sequence of sample times. 
This augmented coalescent method is presented in section~\ref{sec:augco}. 

If sampling at a single time point (homochronously) we show that estimates based on CoMs and BDMs are very similar. In section~\ref{sec:dkb}, we show how the distribution of coalescent times predicted by CoM converges with large sample size to the distribution given by BDM.

\section{The demographic and sampling processes \label{sec:demogProcess}}
	The population size $Y(t)$ is modelled as a continous time Markov chain on $[0,\infty)$, which is governed by the following transition probabilities:
	\begin{align}
		P( Y(t+\Delta t) &= Y(t) + 1) = \lambda Y(t) \Delta t  + O((\Delta t)^2) \notag \\
		P( Y(t+\Delta t) &= Y(t) - 1) = \mu Y(t) \Delta t + O((\Delta t)^2) \\
		P( Y(t+\Delta t) &= Y(t) ) = 1 - (\lambda + \mu) Y(t) \Delta t + O((\Delta t)^2) \notag
	\end{align}
	where $\lambda$ and $\mu$ are the per-capita birth and death rates of the process, respectively. 
	Initially,  $Y(0) = 1$. 
	
	%~ sampling
	We investigated three distinct sampling processes for the reconstruction of genealogies from a simulated demographic history:
	\begin{enumerate}
		\item Continuous sampling through time at constant rate. According to this model, after a lineage dies (with a per-lineage rate $\mu$), it is sampled with independent probability $p$.
		\item Homochronous sampling. According to this model, every extant lineage at a predetermined time point is sampled with independent probability $\psi$. 
		\item Weighted sampling through time at changing rate. According to this model, a weighted sample of $n$ lineages at the time of death is taken with sampling weights that depend on time. If $\{t_i\}$ is the set of death times for lineages indexed by $i$, the sample weights are $w_i = e^{\alpha t_i}$. 
	\end{enumerate}
	Note that with the exception of homochronous sampling, the lineages are only sampled at the time of death. This design is chosen for mathematical convenience, since it eliminates the possibility of a zero-branch length in the genealogy, but both BDMs and CoMs can be generalized to this situation. 
	
	%~ det approx
	In this manuscript, we use CoMs based on the following deterministic approximation $y(t)$ to the stochastic process $Y(t)$:
	\begin{align}
	y(t) &= y(0) e^{t(\lambda - \mu)}, 
	\end{align}
	with real-valued initial conditions $y(0)$ that will be estimated. 
	
	Genealogies were generated by simulation of the birth-death process in continuous time using the software MASTER 1.7.1 \cite{vaughan2013stochastic}.
	For simulating genealogies with a time-dependent sampling rate, we developed a custom simulator for the birth-death process in Python (see supporting information). 
	We simulated 300 genealogies for each of the three sampling scenarios given above using $\mu=1$ and $\lambda=2$ or $\lambda=1.25$. 
	In the case of sampling through time, we terminated the simulation when 100 samples were collected and using a sampling probability of $p=1\%$ or $p=50\%$. 
	If sampling homochronously, we sampled 100 taxa after 9.21 or 25 units of time, yielding a sample proportion that varied around 1\% or 20\%, respectively.
	Simulations that failed to reach the target sample size were removed. 

\section{Estimation methods}
	All models are fitted by maximum likelihood (ML). The choice of ML was motivated by the simplicity of the demographic process, the small number of free parameters, and the possibility that arbitrary priors in a Bayesian framework might obfuscate the important differences between methods.
	For the exponential growth process, there are four potential parameters that could be estimated: birth rates $\lambda$, death rates $\mu$, the initial population size $y(0)$ (needed for CoMs but not BDMs), and a parameter that describes sampling (needed for BDMs but not CoMs). 
	As previously shown in the analysis of BDMs, at most two of these parameters are identifiable from a genealogy alone, and we must therefore choose which parameters to fix according to prior knowledge, and CoMs are subject to the same identifiability constraints.
	We focus on an epidemiologically plausible scenario, where birth rates and the number of infections are unknown, but independent clinical information provides information on death rates.
	Consequently, we will assume $\mu=1$ is known, and will focus on the estimation of birth rate $\lambda$ along with the nuisance parameters describing initial population size (for CoMs) or sampling rates (for BDMs). We will also consider the special case of homochronous sampling, in which we can reparameterise the CoM such that, like the BDM, estimates of the sampling fraction can be obtained. 
	
	Throughout the remainder of the paper, we will use two symbols to denote time on different axes, and all dynamic variables will be defined on both axes. $t$ will denote time from an arbitrary point in the past, while $s$ will denote time before present. It will be useful to define the population genetic models in terms of the retrospective time axis $s$. 
	
	Let $\mathcal{G} = (\mathcal{N}, \mathcal{E}, X)$ represent a genealogy consisting of a set of nodes $\mathcal{N}$, edges $\mathcal{E}$ and a function $X:\mathcal{N}\rightarrow \mathds{R}$ that gives the time $s$ before the present of each node. Every edge corresponds to a 2-tuple $(u,v)$ such that $u,v \in \mathcal{N}$ and the node $u$ is said to be ancestral to $v$. We will consider only rooted binary genealogies; every internal node has exactly two descendents, and all internal nodes but the root have exactly one ancestor.
	
	For CoMs, we use the likelihood given in \cite{volz2012complex}, and we denote the MLE birth rate $\hat{\lambda}^C$. \\
	This likelihood is that of a time-inhomogenous point process with a hazard rate that depends on the population size and number of extant lineages. 
	Specifically, following the approach in~\cite{volz2012complex}, the total population birth rate will be denoted $f(s) = \lambda y(s)$ and the rate of coalescence is 
	\begin{align}
	r(s) &= f(s) \frac{{A(s) \choose 2}}{{y(s) \choose 2}}  \notag  \\
	     &=  {A(s) \choose 2} \frac{2\lambda}{y(s)-1}, 
	\end{align}
	where the term on the right can be understood as the hypergeomatric probability of selecting two lineages that are ancestral to the sample out of the set of $y(s)$ lineages. 
	Now let $x'$ denote the vector of node times (including sampled tips) in ascending order. 
	The likelihood of the $i$'th interval is 
	\begin{align}
	L_i &= 
		\begin{cases}
		e^{-\int_{x_i}^{x_{i+1}} r(s) \mathrm{d} s} & x_{i+1}\textrm{ is a sample time}\\
		r(x_{i+1}) e^{-\int_{x_i}^{x_{i+1}} r(s) \mathrm{d} s} & x_{i+1}\textrm{ is a coalescent time}\\
		\end{cases}
	\end{align}
	And the likelihood is 
	\begin{align}
		\mathcal{L}_{com}(\mathcal{G} | \lambda, \mu, y_0) &= \prod_{i=0}^{2n-2} L_i
	\end{align}
	Note that the number of terms in the likelihood is the number of internode intervals $2n-2$ if all sampling times are distinct. One subtlety arises if more than one lineage is sampled at a single time point, such as with a homochronous sample, in which case we simply deduplicate elements in the vector $x'$ and adjust the number of terms in the likelihood. 
	
	For BDMs, we used the ML framework described in \cite{stadler2010sampling}. %, equation 3 \&  and the maximum likelihood method in section 4.3. 
	We denote the MLE birth rate $\hat{\lambda}^{BD}$.
	The R package \emph{expoTree} \cite{leventhal2014using} was used along with the implementation described here, and all results presented below are based on the best performing of the two implementations of the BDM likelihood. 
	We simplified the likelihood equations in~\cite{stadler2010sampling} to two situations: sampling occurs at a single time point with sample fraction $\rho$, or individuals are sampled with probability $p$ at the time of death. 
	%~ obj_fun.2 <- function(logitrho, lnbeta)
			%~ { #version in MBE paper (similar to eqn 3 in jtb) and following the ML approach in Sec 4.3 of jtb
				%~ beta <- exp(lnbeta)
				%~ rho <- 1 / (1 + exp(-logitrho))
				%~ delta <- true_parameters$delta
				%~ c1 <- abs(beta - delta)
				%~ c2 <- -(beta - delta - 2 * beta * rho) / c1
				%~ .ooq <- function(xi) {
				  %~ o <- 1 / (2*(1-c2^2) + exp(-c1*xi)*(1-c2)^2 + exp(c1*xi)*(1+c2)^2 )
				  %~ #if (is.infinite(o)) browser()
				  %~ o
				  %~ }
				%~ t3 <-  sum(sapply(xis, function(xi) log(.ooq(xi)))) 
				%~ t1 <- log(beta^(n-1) * (4 * rho)^n)
				%~ t2 <- log( mean( sapply( x0s, .ooq ) )   )
				%~ o <- -t1 - t2 - t3
				%~ ifelse(is.infinite(o), 9999999, o)
			%~ }
	Let $x$ denote the vector of times before present for each internal node in $\mathcal{G}$ in descending order. 
	Note that $x_0$ corresponds to the root of the tree. 
	If the sampling takes place according to the homochronous process, $\rho$ will denote the probability of sampling a lineage at a single point in time. 
	Then,
	\begin{align} \label{eqn:bdm1}
	\mathcal{L}_{bdm}(\mathcal{G} | \lambda, \mu, \rho) &= \lambda^{n-1} (4 \rho)^n \prod_{i=0}^{n-2} q(x_i,c_2)^{-1} \left( \int_{x_{or}=x_0}^\infty q(x_{or}, c_2)^{-1}  \mathrm{d} x_{or} \right), 
	\end{align}
	where $q(\cdot)$ is derived from the birth-death sampling formuli: 
	\begin{align}
		q(s,c) &=  2 (1-c^2) + e^{-c_1 s} (1-c)^2 + e^{c_1 s} (1+c)^2, 
	\end{align}
	and $c_1$ and $c_2$ are the following constants:
	\begin{align}
	c_1 &= | \beta - \delta |  \\
	c_2 &= -(\beta - \delta - 2 \beta \rho) / c_1
	\end{align}
	Note that the integral in the likelihood equation is an attempt to account for the unobserved time of origin of the birth-death process.

	%~ obj_fun <- function(lnbeta, lnp)
			%~ { #version in MBE paper (similar to eqn 3 in jtb) and following the ML approach in S 4.3 of jtb
				%~ beta <- exp(lnbeta)
				%~ p <- exp(lnp)
				%~ delta <- true_parameters$delta
				%~ c1 <- sqrt((beta - delta)^2 + 4*beta*delta*p)
				%~ c2_rho0 <- -(beta - delta) / c1 #
				%~ c2_rho1 <- -(-delta - beta) / c1 #
				%~ likterm <- function( c2, x0 )
				%~ {
					%~ .p0 <- function(t) (beta + delta + c1 * (exp(-c1*t)*(1-c2) - (1+c2) )/(exp(-c1*t)*(1-c2)+(1+c2)) ) / (2 * beta)
					%~ .q <- function(t) 2*(1-c2^2) + exp(-c1*t)*(1-c2)^2 + exp(c1*t)*(1+c2)^2
					%~ C <- (delta * p * beta)^n / ( beta * .q(x0) )
					%~ C * prod(sapply( tY, function(t) 1/.q(t) )) * prod(sapply( tZ, function(t) .q(t) ))
				%~ }
				%~ # NOTE  * (1-.p0(x0) ) will be left out, since it is not in the 2012 MBE paper
				%~ #note p0 does not enter into product over tZ since we are assuming no transmissions after sampling
				%~ -log( mean( sapply( x0s, function(x0)  beta * ( likterm(c2_rho0, x0) - likterm(c2_rho1, x0)) ) )  )
			%~ }
	If sampling heterochronously at rate $\psi = \mu p$, the likelihood has a different form. 
	Let $x$ denote the vector times before present of each node as above, and let $y$ denote the vector of sample times in any order. 
	\begin{equation}
	\begin{split}
	\mathcal{L}_{bdm}(\mathcal{G} | \lambda, \mu, \psi) &= 
		 \int_{x_{or}=x_0}^\infty  \lambda^{n-1} \psi^n \{~ q(x_{or, c_2})^{-1} \prod_{i=0}^{n-1} q(y_i, c_2) \prod_{i=0}^{n-2} q(x_i, c_2)^{-1}   \\
		  &-  q(x_{or}, c_3)^{-1} \prod_{i=0}^{n-1} q(y_i, c_3) \prod_{i=0}^{n-2} q(x_i, c_3)^{-1} \} \mathrm{d} x_{or}, 
	\end{split}
	\end{equation}
	and, $c_1$ and $c_2$ and $c_3$ are the following constants:
	\begin{align}
	c_1 &= \sqrt{(\lambda-\mu)^2 + 4 \lambda \psi} \\
	c_2 &= ( \mu - \lambda ) / c_1  \\
	c_3 &= (\lambda + \mu) / c_1
	\end{align}
	
	BDMs and CoMs were fitted according to the same numerical algorithm, with maximization of the likelihood accomplished in R using the simplex method. In order to assure convergence to the global maximum, multiple starting conditions were drawn from a multivariate uniform distribution, and the likelihood optimized for each. The best model fit is reported among the three or five optimizations, although in general they converged to the same value.
	
	\subsection{Estimating birth rates using times of sampling \label{sec:tos}}
		Consider the sequence of sample times in increasing order $\vec{t} = (t_1, \cdots, t_n)$. 
		If the sampling process is known, the sequence of sample times are informative about population size.
		We will consider a simplistic sampling process such that individuals are sampled at a constant rate upon death, which is the sampling process underlying current BDMs. 
		If sampling occurs at a constant known rate, it is straightforward to estimate the historical population size from the sample times, since the probability that a sample will be observed at some point in time is proportional to population size at that time. 
		Therefore, it is possible to estimate the population size using sample time information alone, and not using the genealogy. 
		We show here that it is possible to estimate the birth rate, even if the sampling rate is unknown. 
		Two simple estimators for are presented. The first is based on a simple regression with the expected cumulative number of samples through time. 
		The second is based on treating the sample times as arising from a point process and using maximum likelihood. 
		%~ calculate.regression.Lt <- function()
		%~ { # fit the sample time likelihood model, estimate beta
			%~ Lt_fits2 <- c()
			%~ for (itree in WHICHTREES)
			%~ {
				%~ print(paste(date(), 'starting tree', itree))
				%~ tree <- trees[[itree]]
				%~ sampleTimes <-  .calculate.sample.times(tree) 
				%~ sortSampleTimes <- sort(sampleTimes)
				%~ Lt_fits2 <- rbind(Lt_fits2, coef(lm(  log( 1:length(sortSampleTimes) )~ sortSampleTimes)) )
			%~ }
			%~ colnames(Lt_fits2) <- c( 'intercept', 'r1')
			%~ save( Lt_fits2, file=paste('BD-exponential0-LtFits2.RData', sep='') )
		%~ }
		
		Let $S(t)$ denote the cumulative number of samples collected up to time $t$. 
		We show that the cumulative number of samples increase at the same rate as the unknown population size.
		According to the deterministic model, the expected change in $S$ over time $\Delta t$ will be
		\begin{align}
		\Delta S(t) &= (\Delta t) p \mu y(t) + O((\Delta t)^2) \notag \\
		 &= (\Delta t) p \mu y(0) e^{(\lambda-\mu)t} + O((\Delta t)^2) 
		\end{align}
		Consequently, the logarithm of $S(t)\propto (\lambda-\mu) t$. 
		Regressing the vector $ \log(S(\vec{s}) $ on the vector $\vec{s}$ yields an estimate of the growth rate
		$k = \lambda-\mu $, and using knowledge of $\mu=1$ we have the regression estimator
		\begin{align}
		\hat{\lambda}^R  = \hat{k} + \mu
		\end{align} 
		
		Figure \ref{fig:regression_lt} shows the number of samples through time, for a single simulated genealogy along with the regression line. %These increase exponentially at the same rate as the population size. 
		\begin{figure}
		{ \centering  \includegraphics[width=.7\textwidth]{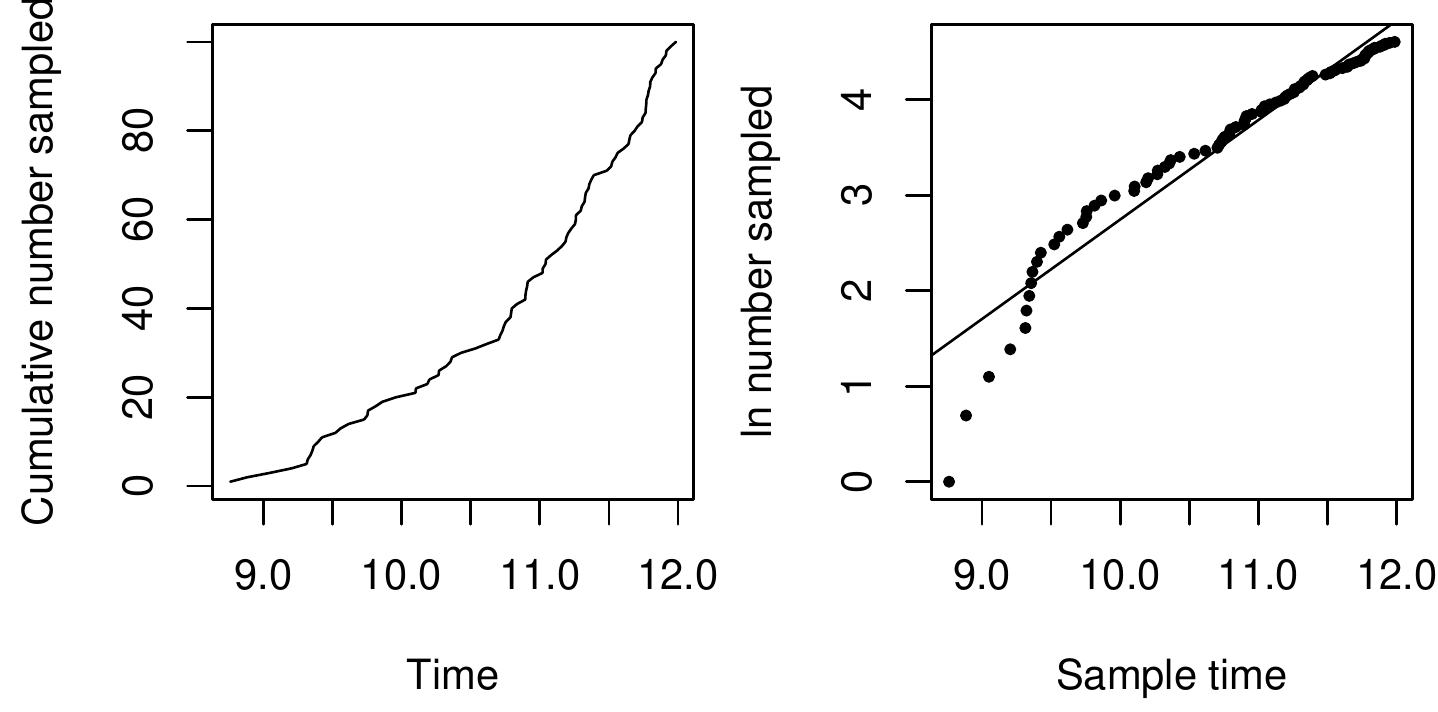}  }  
		\caption{Left: Cumulative number of samples through time. Right: Log cumulative samples with regression line.  \label{fig:regression_lt}}
		\end{figure}
		
		%~ LIk method
		The likelihood approach is based on modeling the sequence of sample times as a point process and also makes use of the deterministic approximation to population size. 
		The rate of a sample appearing at time $t$ is 
		\[
		f(t) = p \mu y(0) e^{(\lambda-\mu)t} \\
		= a e^{k t},
		\]
		with $a = p \mu y(0)$ and $k = \lambda - \mu$.
		The probability of $\vec{s}$ is 
		\[
		P(\vec{t} | a, k) = \prod_{i=1} f(t_i) e^{-\int_{t_{i-1} }^{t_i} f(\tau) \mathrm{d} \tau}, 
		\]
		with $s_0 = 0$ for consistency. 
		As with the regression estimator, 
		\begin{align}
		\hat{\lambda}^S  = \hat{k} + \mu
		\end{align}
	
	\subsection{The augmented coalescent model \label{sec:augco}}
		The genealogy $\mathcal{G}$ and the sample times $\vec{s}$ are conditionally independent given demographic and sampling parameters $\theta=(\lambda,\mu,p, y(0))$. Therefore, the likelihood of both is the product of the marginal likelihoods given above:
		\begin{align} \label{eqn:augco}
		P(\mathcal{G},\vec{t} | \theta) = P(\mathcal{G} | \theta) P(\vec{t} | \theta )
		\end{align}
		We will denote the MLE birth rate as $\hat{\lambda}^A$. 

\section{Results}
	The following results demonstrate the level of bias, precision and efficiency of different inference methods when estimating the birth rate from genealogies generated by the birth death process. 
	\subsection{Constant sampling rate}
		Figure \ref{fig:vioplot-exponential0} shows the distribution of MLEs for five estimators presented above based on 300 simulated genealogies. 
		Simulations were based on a sampling process with constant sampling probability $p=1\%$ at the time of death. 
		
		Estimators based only on the sequence of sample times $\vec{s}$ perform well even though they do not use the coalescence times. 
		The ML estimator $\hat{\lambda}^S$ consistently outperforms the simple regression estimator $\hat{\lambda}^R$, presumably reflecting that residuals in the loglinear regression model are not i.i.d. normal. 
		
		Comparing a model that only uses sample time information ($\lambda^S$) with the CoM that only uses genealogical information ($\lambda^C$) we find that that the RMSE of $\hat{\lambda}^S$ is 11.5\% 
		compared to 11.8\% for $\hat{\lambda}^C$. 
		In this instance, the sample time sequence is actually more informative than the coalescent times for inferring birth rates. 
		
		Comparing the BDM and coalescent, we find that the BDM is more precise (RMSE = 0.085) but slightly less accurate; the average bias of the BDM estimator was 0.022 (95\% CI:(0.013, 0.031) ) compared to 0.009(95\% CI: (-0.022, 0.005)) for the CoM estimator. 
		Comparing the BDM and augmented CoM (a model that uses both coalescence and sample times), we find that the augmented CoM is slightly less precise than the BDM (RMSE of $\hat{\lambda}^A$ is 0.092), which may reflect the use of a misspecifed deterministic population size, however, we did not detect significant bias of $\hat{\lambda}^A$ (95\% CI:(-0.012,0.009)) in contrast to the BDM.

% > ltbetas <- .Lt.accuracy(Lt_fits)
% [1] "Lt accuracy"
% [1] "bias"
% [1] 0.01574863
% [1] "rmse"
% [1] 0.1151333
% > print('')
% [1] ""
% > ltbetas2 <- .Lt.accuracy2(Lt_fits2)
% [1] ""
% [1] "Lt accuracy 2, regression estimator"
% [1] "bias"
% [1] -0.05471234
% [1] "rmse"
% [1] 0.1433589
% [1] ""
% > print('')
% [1] ""

% coalescent
% [1] "bias"
% [1] -0.008652696
% [1] "rmse"
% [1] 0.1182815
% [1] "variance of betas"
% [1] 0.01396219

% 	One Sample t-test

% data:  betas - 2
% t = -1.2683, df = 299, p-value = 0.2057
% alternative hypothesis: true mean is not equal to 0
% 95 percent confidence interval:
%  -0.022078049  0.004772657
% sample estimates:
%   mean of x 
% -0.008652696 

% [1] ""
% [1] "bd accuracy"
% [1] "betas"
% [1] "bias"
% [1] 0.02204957
% [1] "rmse"
% [1] 0.08236811
% [1] "~~~~~~~"
% [1] "p"
% [1] "bias"
% [1] 0.0007689401
% [1] "rmse"
% [1] 0.003125557
% [1] "variance of betas"
% [1] 0.006319386

% 	One Sample t-test

% data:  betas - 2
% t = 4.8042, df = 299, p-value = 2.462e-06
% alternative hypothesis: true mean is not equal to 0
% 95 percent confidence interval:
%  0.01301751 0.03108162
% sample estimates:
%  mean of x 
% 0.02204957 

% [1] ""
% [1] "aug co accuracy"
% [1] "bias"
% [1] -0.001264785
% [1] "rmse"
% [1] 0.09189584
% [1] "variance of betas"
% [1] 0.008471485

% 	One Sample t-test

% data:  betas - 2
% t = -0.238, df = 299, p-value = 0.812
% alternative hypothesis: true mean is not equal to 0
% 95 percent confidence interval:
%  -0.011722305  0.009192735
% sample estimates:
%   mean of x 
% -0.001264785 

% [1] ""

		\begin{figure}
		{ \centering  \includegraphics[width=.8\textwidth]{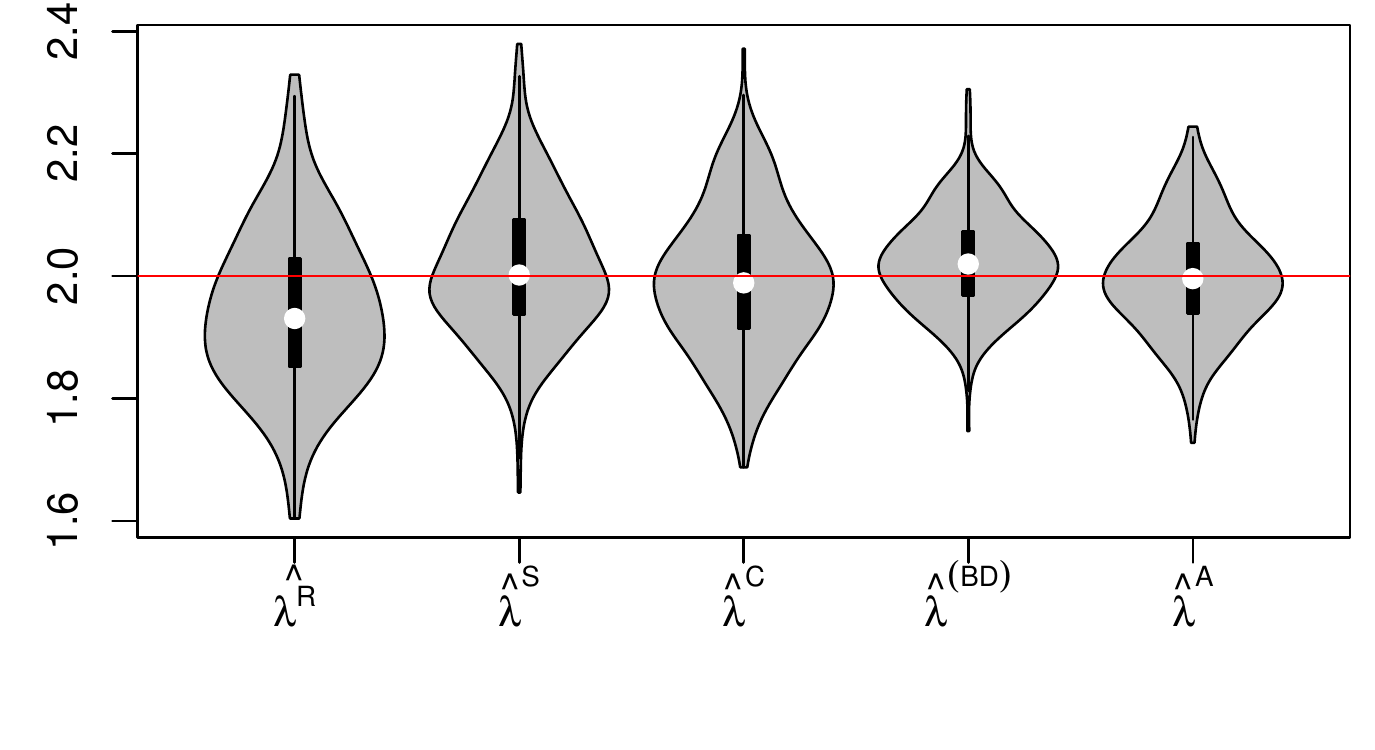} }
		\caption{ Distribution of MLE birth rates from 300 simulations with constant sampling rate and using five different estimators. The red line shows the true parameter value. $\hat{\lambda}^R$ is a regression estimator, $\hat{\lambda}^S$ is a ML estimator using sample times, $\hat{\lambda}^C$ is the coalescent estimator, $\hat{\lambda}^{(BD)}$ is the BDM estimator, and $\hat{\lambda}^A$ is the coalescent estimator that also uses sample times. \label{fig:vioplot-exponential0} }
		\end{figure}
		
		%~ lik pair plot
		\begin{figure}
		{ \centering \includegraphics[width=.8\textwidth]{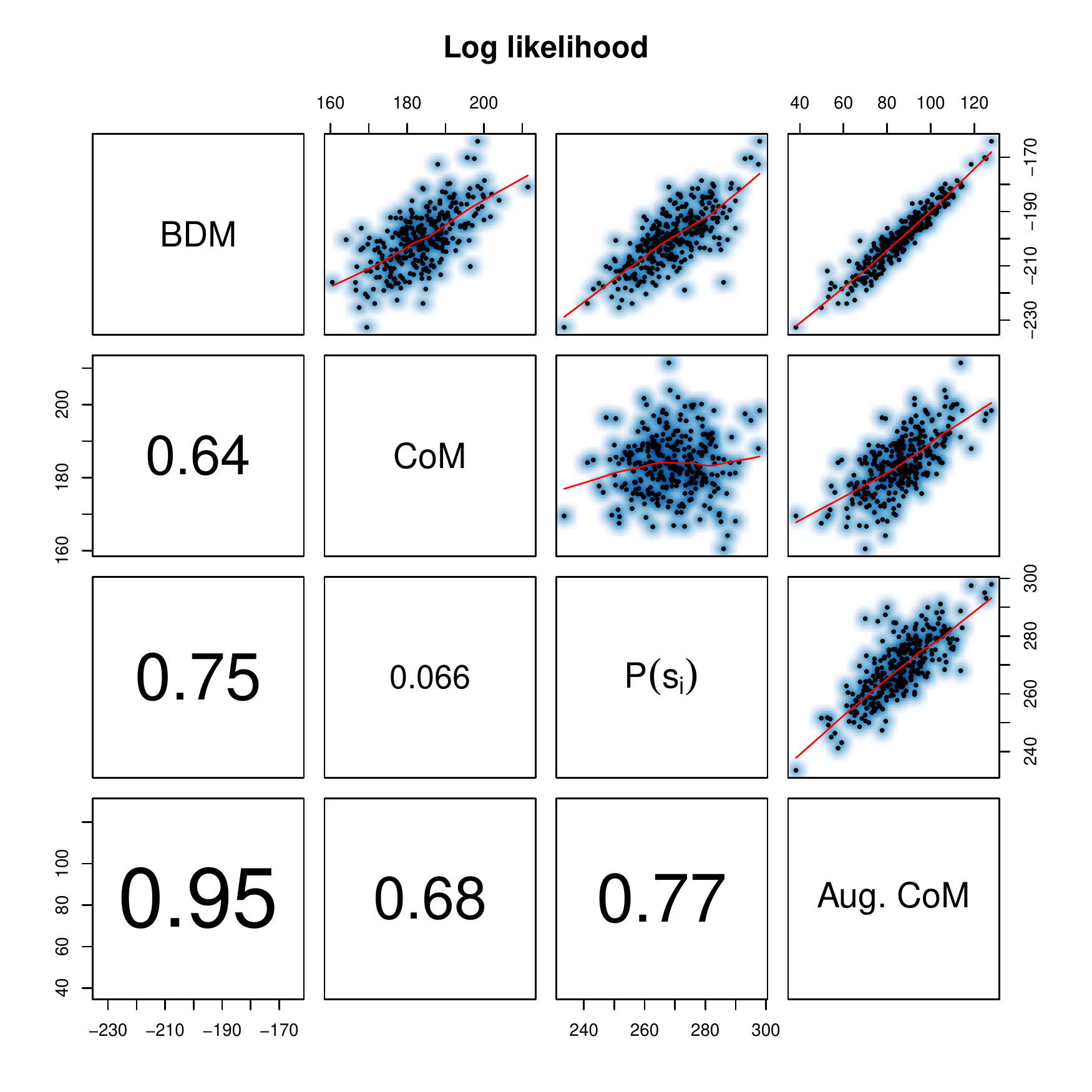}  }
		\caption{ The log likelihood corresponding to MLEs from four estimation methods and based on 300 simulated genealogies. The Pearson correlation coefficient between log likelihoods is shown in the lower panels. The upper panels show a scatter plot with smoothing splines.   \label{fig:likpairs}}
		\end{figure}

		Figure~\ref{fig:likpairs} sheds some light on why the estimators perform differently by comparing the maximum likelihood estimated by each method on each simulated genealogy. 
		SI figure~\ref{fig:lambdapairs} shows a similar scatter plot of MLE birth rates. 
		The BDM likelihood is highly correlated with that of all other estimators. 
		In contrast, the CoM likelihood is almost independent of the estimators that use sample times only ( Pearson correlation $= 0.066$). 
		The highest correlation is found between the BDM and the augmented CoM (Pearson correlation $=0.95$). 
		This illustrates that the CoM is not using sample time information, but the BDM and augmented CoM are using information from both the sample times and genealogy. 
	
	\subsection{Homochronous sampling \label{sec:homochron}}
		If all samples are collected at a single point in time, and if the sampling proportion is unknown, then the time of sampling and sample size confer no information about population size. 
		The homochronous sampling case with unknown sampling rate therefore provides a fair comparison for BDMs and CoMs. 
		Here we consider 300 simulations of the birth death process with a sample of $n=100$ at $t=9.2$, so that the sample fraction is around 1\%, though it differs between replicates.
		The birth rate used in the simulations was $\lambda=2$. 
		
		Figure~\ref{fig:vioplot-homochronous0} shows the distribution of MLE birth rates. 
		The distributions are very similar and have similar precision (RMSE of $\hat{\lambda}^{BD}$ is 0.106 and RMSE of $\hat{\lambda}^C$ is 0.101). 
		The CoM estimator does not have detectable bias (95\% CI of bias: (-0.0183,0.0048) ), but the BD model slightly overestimates birth rates (average bias = 0.036, 95\% CI: (0.0242,0.0470)). %
		Figure~\ref{fig:vioplot-homochronous0} also shows that the log likelihoods of the MLEs generated by both methods are highly concordant up to a constant factor. The Pearson correlation of BDM and CoM maximum likelihoods is 99.6\%. 
		The estimated birth rates also have a high correlation coefficient of 86.6\%. 
		%~ 0.9956963
		
			%~ data:  bdbetas - 2
			%~ t = 6.1546, df = 297, p-value = 2.433e-09
			%~ alternative hypothesis: true mean is not equal to 0
			%~ 95 percent confidence interval:
			 %~ 0.02423045 0.04701012
			%~ sample estimates:
			 %~ mean of x 
			%~ 0.03562028 
			%~ 
			%~ 
				%~ One Sample t-test
			%~ 
			%~ data:  cobetas - 2
			%~ t = -1.1545, df = 297, p-value = 0.2492
			%~ alternative hypothesis: true mean is not equal to 0
			%~ 95 percent confidence interval:
			 %~ -0.01830039  0.00476786
			%~ sample estimates:
			   %~ mean of x 
			%~ -0.006766268 
					
				%~ #~ [1] "RMSE"
		%~ #~        co        bd     mifco 
		%~ #~ 0.1016540 0.1065541 0.1078468 
		%~ #~                co          bd       mifco
		%~ #~ co    0.010300729 0.008847310 0.008382087
		%~ #~ bd    0.008847310 0.010137128 0.009728172
		%~ #~ mifco 0.008382087 0.009728172 0.010541915
		%~ #~              co        bd     mifco
		%~ #~ co    1.0000000 0.8658044 0.8043747
		%~ #~ bd    0.8658044 1.0000000 0.9410528
		%~ #~ mifco 0.8043747 0.9410528 1.0000000

		\begin{figure}
		{ \centering \includegraphics[width=.8\textwidth]{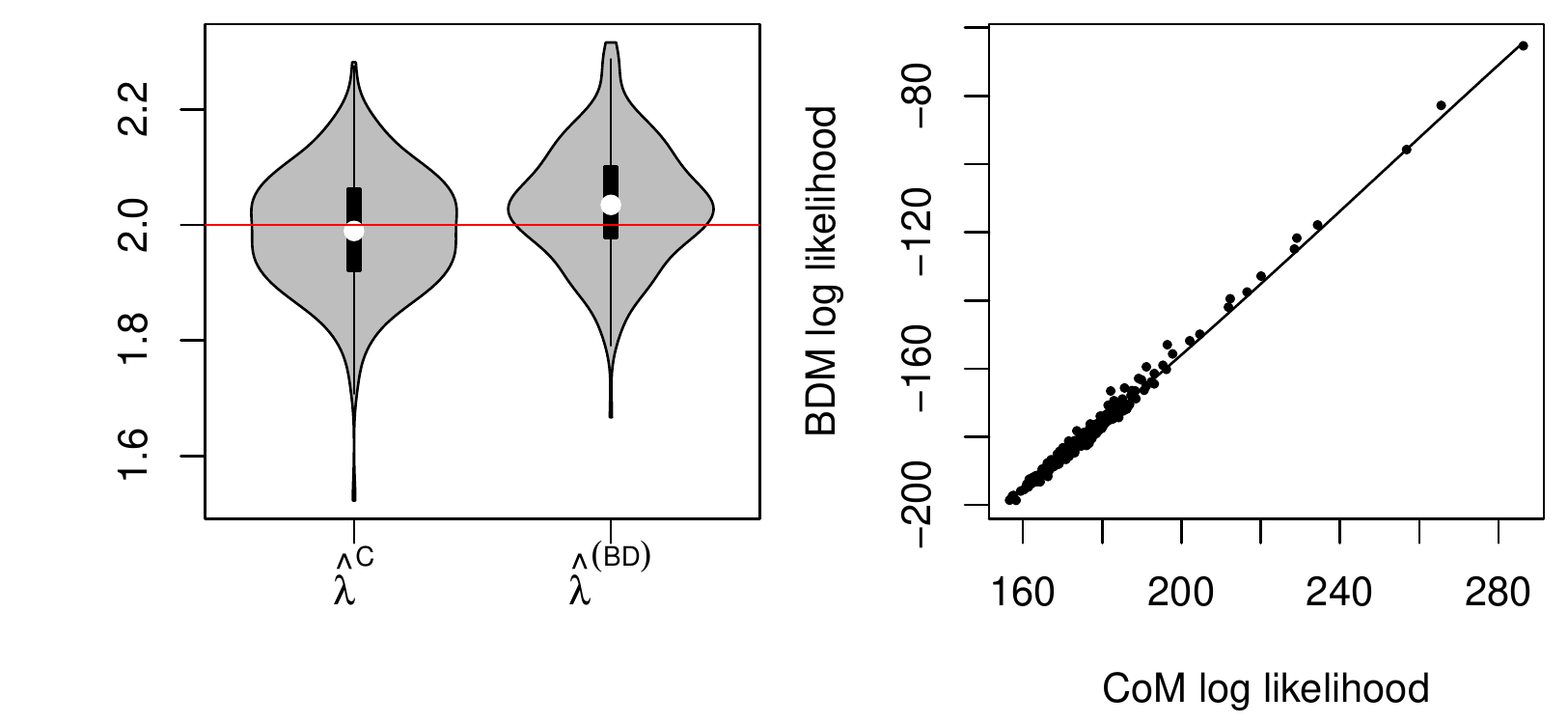} }
		\caption{ MLEs based on simulations with homochronous sampling. Left: Distribution of MLE birth rates from 300 simulations and using BDM and CoM estimators. The red line shows the true parameter value.
		Right: The log likelihoods of the BDM and CoM MLEs with a smoothing spline. 
		 \label{fig:vioplot-homochronous0} }
		\end{figure}
		
		%~ calc percent variance explained by sample times vs genealogy
		Comparing the RMSE of the BDM estimator in both the homochronous and constant sampling rate cases, it appears that having informative sample time information decreases the residual sums-of-squares of the BDM estimator by about 36\%, but this gain in precision will certainly depend on parameters of the system and sample size.

		% homochronous1 experiment with lambda=1.25
		We repeated the simulation exercise with a smaller birth rate ($\lambda=1.25$) in order to assess if the CoM estimator would be less accurate if the population is growing more slowly. 
		The MLEs are depicted in SI figure~\ref{fig:vioplot-homochronous1}. 
		With the smaller birth rate, we do not detect significant bias of the BDM estimator (average bias $<1e-3$, 95\% CI:$(-0.0069,0.0077)$), or with the CoM estimator (average bias 0.002, 95\% CI of bias: $(-0.0057,0.0096)$). 
		The RMSE of the BDM and CoM estimators are similar (0.037 and 0.039 respectively). 
		
		%~ #~ 	One Sample t-test
		%~ #~ 
		%~ #~ data:  bdbetas - true_parameters[["beta"]]
		%~ #~ t = 0.1044, df = 105, p-value = 0.917
		%~ #~ alternative hypothesis: true mean is not equal to 0
		%~ #~ 95 percent confidence interval:
		%~ #~  -0.006914331  0.007682948
		%~ #~ sample estimates:
		%~ #~    mean of x 
		%~ #~ 0.0003843082 
		%~ #~ 
		%~ #~ 
		%~ #~ 	One Sample t-test
		%~ #~ 
		%~ #~ data:  cobetas - true_parameters[["beta"]]
		%~ #~ t = 0.5103, df = 105, p-value = 0.6109
		%~ #~ alternative hypothesis: true mean is not equal to 0
		%~ #~ 95 percent confidence interval:
		%~ #~  -0.005658371  0.009580485
		%~ #~ sample estimates:
		%~ #~   mean of x 
		%~ #~ 0.001961057 
		%~ #~ 
		%~ #~ 
		%~ #~ 
		%~ #~ [1] "RMSE"
		%~ #~         co         bd      mifco 
		%~ #~ 0.03942509 0.03772045 0.03855368 
		%~ #~ 

	\subsection{Comparison of estimated sample rates}
	    An alternative parameterization of the coalescent is in terms of the population size at the time of sampling in a homochronous scenario. 
	    In this case, we can calculate a deterministic approximation to the population size at time $s$ in the past as 
	    $$ y(s) = \frac{n}{\rho} e^{-s(\lambda - \mu)}, $$
	    where $n$ is the sample size and $\rho$ is the sample proportion, and $n/\rho$ is the population size at the time of sampling.
	    According to this parameterization, we replace the nuisance parameter $y(0)$ with $\rho$, and the coalescent estimates of the sample proportion can be directly compared to estimates with the BDM.
	    
	    We fit the reparameterized CoM to the the same genealogies used in section~\ref{sec:homochron} with $\lambda=1.25$ and $\mu=1$. 
	    We see that the estimates are highly concordant with Pearson correlation of 99.7\%. 
	    
	    \begin{figure}
	    \centering
	    \includegraphics[width=.5\textwidth]{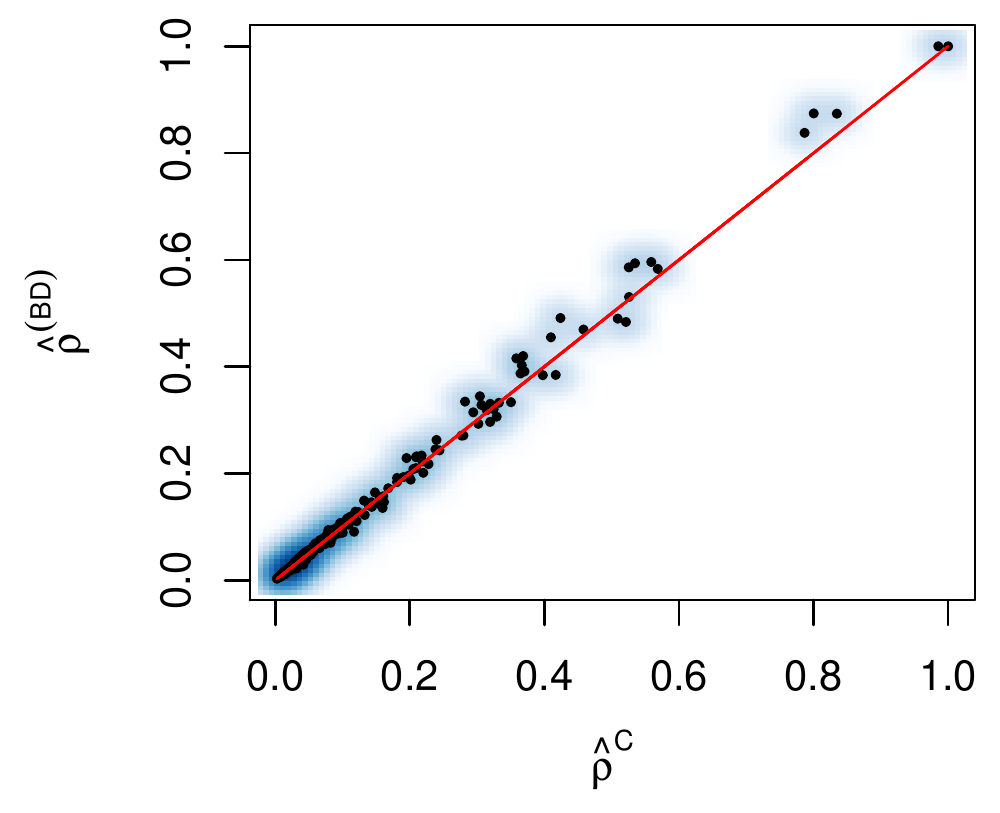}
	    \caption{Estimated sample proportions using the coalescent and BDM with homochronous sampling. $\lambda=1.25, \mu=1, n = 100$.}
	    \label{fig:rhocomparison}
	    \end{figure}
	    
    \subsection{Small reproduction number and high sample fraction}
        The CoM based on a deterministic demographic process may be most biased when the population size is small and subject to large stochastic fluctuations. 
        We generated 300 trees from the BD process with $\lambda=1.1, \mu=1$ and homochronous sampling with $n=100$ and a variable sample fraction around $50\%$.
        The distribution of MLE birth rates is shown in Figure~\ref{fig:homochron2}.
        
        We found small but significant bias in the estimated birth rates using both BDM and CoM methods. 
        The mean bias of the BDM estimator was $0.019$ (95\% CI: ( 0.0148, 0.0244)), and the mean bias of the CoM was $0.021$ (95\% CI of bias: (0.0160, 0.0265)). 
        The BDM had smaller RMSE (0.043 vs 0.47), and the Pearson correlation of estimated birth rates was 95\%. 
        A comparison of estimated birth rates is shown in SI figure~\ref{fig:homochron2lambdaPairs}.

% #~ [1] "bd accuracy"
% #~ [1] "betas"
% #~ [1] "bias"
% #~ [1] 0.0196071
% #~ [1] "rmse"
% #~ [1] 0.0427258
% #~ [1] "~~~~~~~"
% #~ [1] "rho"
% #~ [1] "bias"
% #~ [1] -14.63544
% #~ [1] "rmse"
% #~ [1] 157.7222
% #~ [1] "variance of betas"
% #~ [1] 0.001446914
% #~ [1] "~~~~~~~~~~~~~~"
% #~ [1] "co accuracy"
% #~ [1] "bias"
% #~ [1] 0.0212192
% #~ [1] "rmse"
% #~ [1] 0.04677892
% #~ [1] "variance of betas"
% #~ [1] 0.001745078
% #~ [1] "~~~~~~~~~~~~~~"
% #~ 
% #~ 	One Sample t-test
% #~ 
% #~ data:  bdbetas - true_parameters[["beta"]]
% #~ t = 8.101, df = 246, p-value = 2.546e-14
% #~ alternative hypothesis: true mean is not equal to 0
% #~ 95 percent confidence interval:
% #~  0.01483990 0.02437429
% #~ sample estimates:
% #~ mean of x 
% #~ 0.0196071 
% #~ 
% #~ 
% #~ 	One Sample t-test
% #~ 
% #~ data:  cobetas - true_parameters[["beta"]]
% #~ t = 7.9831, df = 246, p-value = 5.458e-14
% #~ alternative hypothesis: true mean is not equal to 0
% #~ 95 percent confidence interval:
% #~  0.01598381 0.02645459
% #~ sample estimates:
% #~ mean of x 
% #~ 0.0212192 
% #~ 
% #~ [1] "RMSE"
% #~         co         bd 
% #~ 0.04677892 0.04272580 
% #~             co          bd
% #~ co 0.001745078 0.001506654
% #~ bd 0.001506654 0.001446914
% #~          co       bd
% #~ co 1.000000 0.948167
% #~ bd 0.948167 1.000000

        \begin{figure}
        \centering
        \includegraphics[width=.8\textwidth]{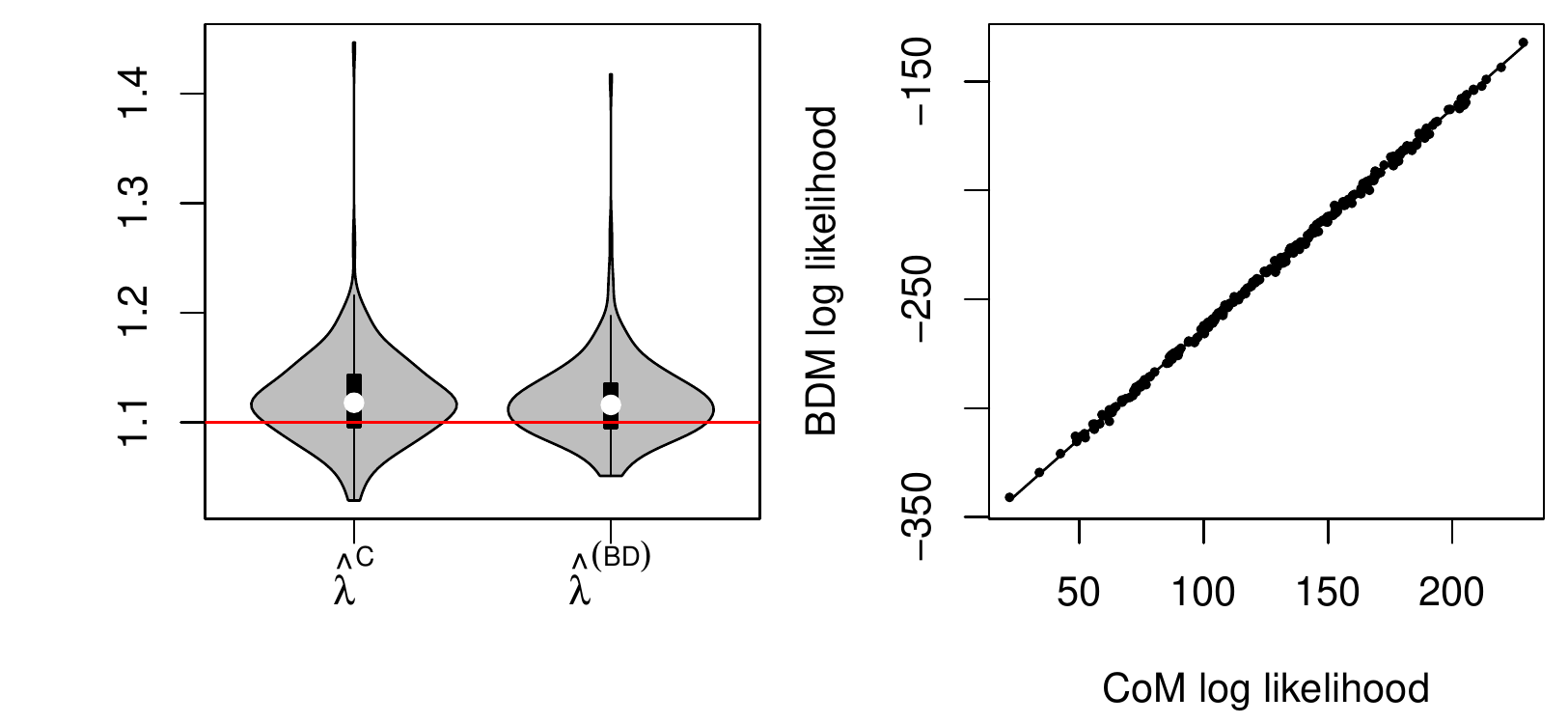}
        \caption{MLEs based on simulations with homochronous sampling, $\lambda=1.1, \mu=1$ and a variable sample fraction around $50\%$. Left: Distribution of MLE birth rates from 300 simulations and using BDM and CoM estimators. The red line shows the true parameter value.
		Right: The log likelihoods of the BDM and CoM MLEs with a smoothing spline. }
        \label{fig:homochron2}
        \end{figure}

	\subsection{Decreasing sample rate and small sample fraction}
		When the sampling model implicit to the BDM approach is misspecified, the BDM may yield highly biased results. 
		Figure~\ref{fig:vioplot-timeDepSampling0} shows the MLE birth rates for both the BDM and CoM estimators when the sampling rate changes through time according to $e^{\alpha t}$ (see section~\ref{sec:demogProcess}).
		120 simulations were carried out, and the sampling rate decreased at a rate of $\alpha=-0.44$. 
		This value was chosen so that the expected sample size would be 100 if taking a weighted sample of all lineages at the time of death. 
		Note that the sampling rate is an exponential function of time, so that the sequence of sample times still appears as though it arises from an exponentially increasing population, and there would be no warning from the sequence of sample times alone that the rate is changing. 
		The BDM estimates are biased downwards by 0.23 (95\% CI: (-0.2488, -0.2207)).
		
		In this scenario, the CoM is robust to changing sample rate, since the CoM conditions on observed sample times. 
		The CoM estimates did not have significant bias (95\% CI of bias: (-0.0443,0.0045)).

			%~ [1] "bd accuracy"
			%~ [1] "betas"
			%~ [1] "bias"
			%~ [1] -0.2347892
			%~ [1] "rmse"
			%~ [1] 0.2472034
			%~ [1] "~~~~~~~"
			%~ [1] "rho"
			%~ [1] "bias"
			%~ [1] NaN
			%~ [1] "rmse"
			%~ [1] NaN
			%~ [1] "variance of betas"
			%~ [1] 0.006033851
			%~ [1] "bias"
			%~ [1] -0.01986405
			%~ [1] "rmse"
			%~ [1] 0.1358322
			%~ [1] "variance of betas"
			%~ [1] 0.01820754
				%~ One Sample t-test
			%~ data:  bdbetas - 2
			%~ t = -33.1109, df = 119, p-value < 2.2e-16
			%~ alternative hypothesis: true mean is not equal to 0
			%~ 95 percent confidence interval:
			 %~ -0.2488301 -0.2207484
			%~ sample estimates:
			 %~ mean of x 
			%~ -0.2347892 
				%~ One Sample t-test
			%~ data:  cobetas - 2
			%~ t = -1.6126, df = 119, p-value = 0.1095
			%~ alternative hypothesis: true mean is not equal to 0
			%~ 95 percent confidence interval:
			 %~ -0.04425463  0.00452653
			%~ sample estimates:
			  %~ mean of x 
			%~ -0.01986405 
			%~ [1] 0.5211481
		\begin{figure}
		{ \centering \includegraphics[width=.8\textwidth ]{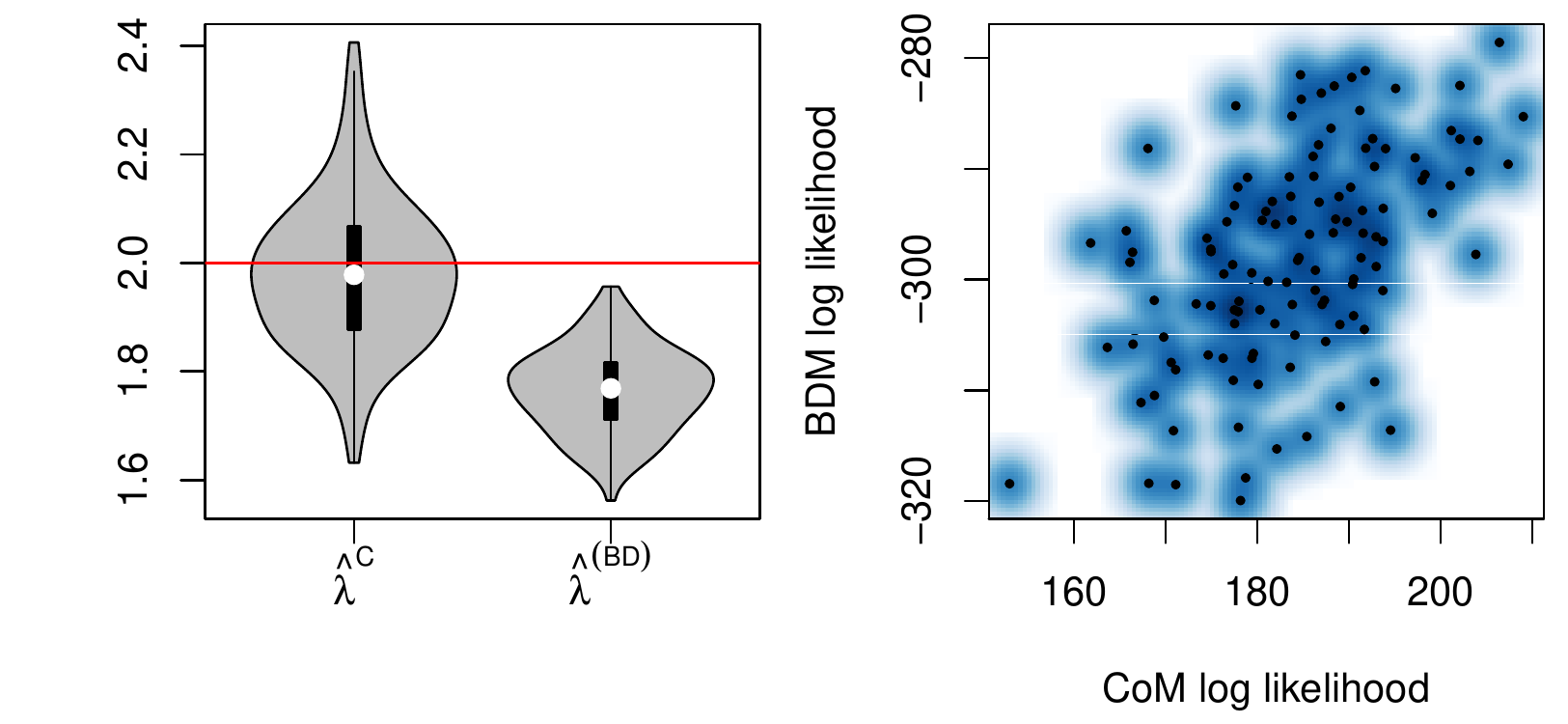} }
		\caption{ MLE birth rates based on simulations with time-dependent sampling. Left: Distribution of MLE birth rates from 120 simulations and using BDM and CoM estimators.
		Right: The log likelihoods of the BDM and CoM MLEs with a smoothing spline.   \label{fig:vioplot-timeDepSampling0}}
		\end{figure}
	
	\subsection{Increasing sample rate and large sample fraction}
		In these experiments, we examine bias in the coalescent due to sampling a large fraction of lineages from a small population growing stochastically. 
		300 genealogies with $n=100$ were simulated from a birth death process.
		Simulations were terminated when the number of deceased lineages reached 200, so that the sample fraction was 50\% of deceased lineages and about 25\% of all lineages. 
		In the same experiments, we exmained bias in BDMs due to a misspecified sampling process. 
		In these experiments, the sampling rate increases from zero at time zero at a rate of $\rho=\mu$. 
		
		Figure~\ref{fig:vioplot-timeDepSampling-highSampFrac} shows the distribution of MLE birth rates. 
		We do not find detectable bias with the CoM estimator (95\% CI:(-0.0271,0.0260)), despite using a misspecified deterministic approximation to the demographic process, and  despite that a large sample of the population was taken and that the population size was only around 400 on average at the time of the last sample.  
		
		Because the BDM relies on a misspecified sampling process, the BDM estimator gives highly biased estimates in this scenario. The average bias was 0.46 (95\% CI:(0.4460,0.4920)). 
		
			%~ One Sample t-test
			%~ 
			%~ data:  bdbetas - 2
			%~ t = 40.125, df = 298, p-value < 2.2e-16
			%~ alternative hypothesis: true mean is not equal to 0
			%~ 95 percent confidence interval:
			 %~ 0.4460196 0.4920266
			%~ sample estimates:
			%~ mean of x 
			%~ 0.4690231 
			%~ 
			%~ 
				%~ One Sample t-test
			%~ 
			%~ data:  cobetas - 2
			%~ t = -0.0394, df = 298, p-value = 0.9686
			%~ alternative hypothesis: true mean is not equal to 0
			%~ 95 percent confidence interval:
			 %~ -0.02706913  0.02600666
			%~ sample estimates:
				%~ mean of x 
			%~ -0.0005312355 

		\begin{figure}
		\includegraphics[width=.5\textwidth]{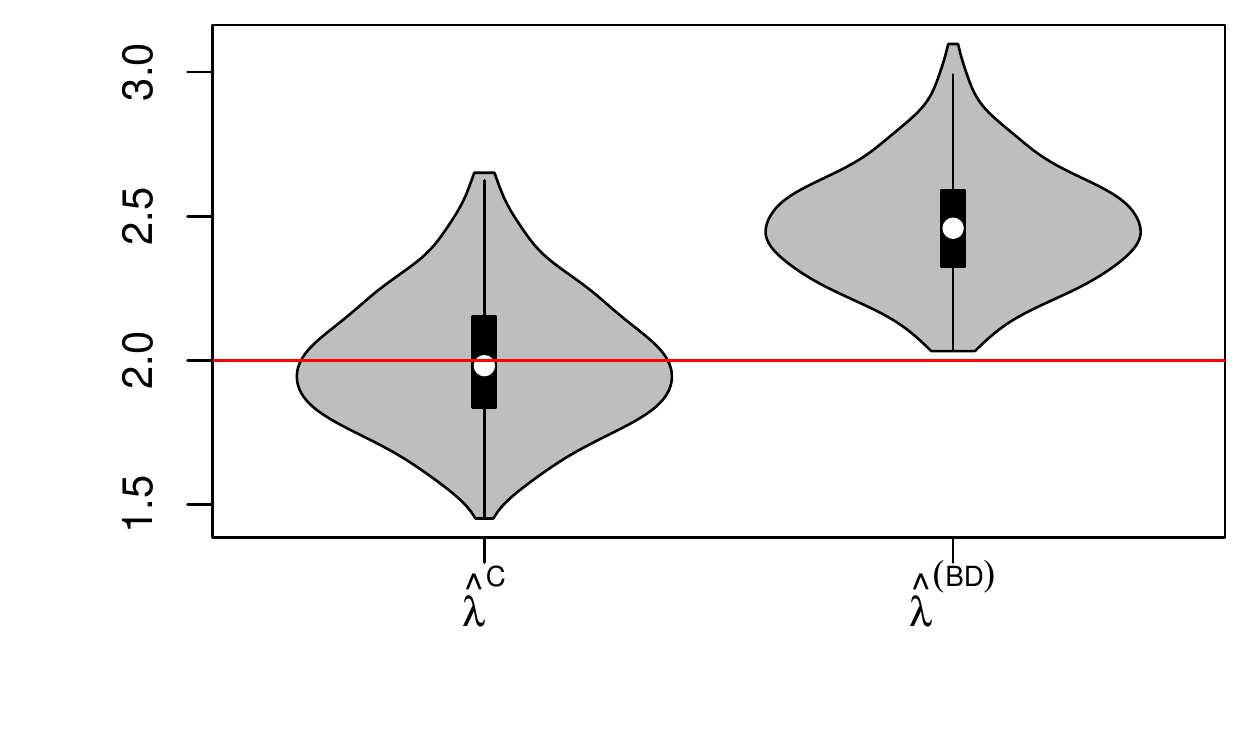}
		\caption{
		MLE birth rates based on simulations with a 50\% sample fraction, $n=100$, and with time-dependent sampling rate that increases through time. Left: Distribution of MLE birth rates from 300 simulations and using BDM and CoM estimators.
		\label{fig:vioplot-timeDepSampling-highSampFrac}}
		\end{figure}

\subsection{Asymptotic distribution of coalescent times \label{sec:dkb}}
    Some insight into why CoM and BDM give similar estimates can be gained by comparing the asymptotic distribution of coalescent times predicted by both models in the case of homochronous sampling. The distribution of coalescent times in the limit of large sample size for a deterministic coalescent model can be easily computed, and we show that this distribution is equivalent to the marginal likelihood of a node given by the birth-death model.
    
    In \cite{frost2010viral,maruvka2011recovering}, an approximation to the lineages through time for the coalescent process was presented for a population under exponential growth:
    \begin{align}
    \frac{\mathrm{d}}{\mathrm{d} s} A &= -{A(s) \choose 2} \frac{2\lambda }{Y(s)} \label{eqn:dads}
    \end{align}
    If sampling occurs at a single time point, such that $A(0)=n$, this has unique solution
    \begin{align}
    A(s) &= \frac{1}{1 - \frac{1}{n} \left(n - 1\right) e^{- \frac{\lambda \left(e^{s \left(\lambda - \mu\right)} - 1\right)}{Y_{0} \left(\lambda - \mu\right)}}},
    \end{align}
    where $Y_0$ is the population size at the time of sampling. 
    We will call this a doubly deterministic coalescent model (DDCoM) because both the demographic and genealogical processes are modeled with deterministic approximations. 
    The asymptotic distribution of coalescent times for the DDCoM is given by the derivative of $A(s)$ (equation \ref{eqn:dads} and expanding $Y(s)$ and normalizing: 
    \begin{align} 
    P_{ddcom}(s | \lambda, \mu, \rho, n)  &= -\frac{\mathrm{d}}{\mathrm{d} s} A / (n-1) \\
     &= \frac{\beta \rho e^{\frac{\beta \rho \left(e^{s \left(\beta - \gamma\right)} - 1\right)}{n \left(\beta - \gamma\right)} + s \left(\beta - \gamma\right)}}{\left(n e^{\frac{\beta \rho \left(e^{s \left(\beta - \gamma\right)} - 1\right)}{n \left(\beta - \gamma\right)}} - n + 1\right)^{2}}. \label{eqn:pddcom}
    \end{align} 
    The factor of $n-1$ normalises the distribution since there are $n-1$ nodes in the tree.
    In \cite{jewett2014theory}, the DDCoM was found to be an excellent approximation to the stochastic coalescent for large populations.

    The BDM likelihood takes the form of a product over coalescent times and sample times, including the time of origin. 
    Conditioning on the time of origin, and given a homochronous sample, the likelihood is given by the product of marginal probabilities for each coalescent time. 
    From equation \ref{eqn:bdm1}, expanding $c_1, c_2$ and simplifying:
    \begin{align}
    P_{bdm}(s | \lambda, \mu, \rho) &= \frac{ 4 \lambda \rho }{q(s, c_2) } \notag  \\
      &= \frac{4 \lambda \rho}{2 (1-c_2^2) + e^{-c_1 s} (1-c_2)^2 + e^{c_1 s} (1+c_2)^2},  \label{eqn:pbdm}
    \end{align}
    where  $c_1$ and $c_2$ are the following constants:
    	\begin{align}
    	c_1 &= | \lambda - \mu |  \\
    	c_2 &= -(\lambda - \mu - 2 \beta \rho) / c_1
    	\end{align}
    
    \begin{theorem}
    	Given a homochronous sample of a proportion $\rho$ lineages from a population growing exponentially according to the birth-death process with birth rate $\beta$, death rate $\gamma$, and $\beta > \gamma$, 
    	$$\lim_{n\rightarrow \infty} P_{ddcom}(s|n, \beta, \gamma, \rho) = P_{bdm}(s|\beta, \gamma, \rho) $$
    	for all times $s$. 
    \end{theorem}
    \begin{proof}
    	By Taylor expansion of the denominator of equation \ref{eqn:pddcom}, we have
    	\begin{align}
    		& \left(n e^{\frac{\beta \rho \left(e^{s \left(\beta - \gamma\right)} - 1\right)}{n \left(\beta - \gamma\right)}} - n + 1\right)^{2}  \notag \\
    		& = \left(1 + \frac{\beta \rho \left(e^{s \left(\beta - \gamma\right)} - 1\right)}{\left(\beta - \gamma\right)}  + O(1/n) \right)^{2}  \label{eqn:p1}
    	\end{align}
    	The limit of the numerator of equation \ref{eqn:pddcom} is 
    	\begin{align}
    	\lim_{n\rightarrow \infty} & \beta \rho e^{\frac{\beta \rho \left(e^{s \left(\beta - \gamma\right)} - 1\right)}{n \left(\beta - \gamma\right)} + s \left(\beta - \gamma\right)} \notag \\
      & = \beta \rho e^{s \left(\beta - \gamma\right)} \label{eqn:p2}
    	\end{align}
    	Taking the large $n$ limit of equation \ref{eqn:p1} and computing the ratio of \ref{eqn:p1} and \ref{eqn:p2},  and rearranging, we have
    	\begin{align}
    	\lim_{n\rightarrow \infty} P_{ddcom}(s|\beta,\gamma, \rho) &=  \frac{\beta \rho (\beta-\gamma)^2 e^{s(\beta-\gamma)}}{\left( \beta - \gamma - \beta \rho + \beta \rho e^{s(\beta - \gamma)}  \right)^2}
    	\end{align}
    	It may be verified that this is equivalent to $P_{bdm}$ (equation \ref{eqn:pbdm}).
    \end{proof}
    
    Note that this result applies to the DDCoM model and not the coalescent model used elsewhere in the text. 
    In \cite{chen2013asymptotic,jewett2014theory} it was shown that the lineages through time given by DDCoMs are generally excellent approximations to lineages through time given by standard CoMs if the sample size is large.

    Outside of the large-$n$ limit, we can investigate the similarity of $P_{bdm}$ and $P_{ddcom}$ numerically. 
    To summarize the difference between distributions $P_{bdm}$ and $P_{ddcom}$, we compute the Kullback–Leibler divergence:
    $$  D(P_{ddcom}, P_{bdm} | \lambda, \mu, \rho, n) = \int_{s=0}^\infty \log(\frac{P_{ddcom}(s|\lambda, \mu, \rho, n) }{P_{bdm}(s | \lambda, \mu, \rho)}) P_{ddcom}(s|\lambda, \mu, \rho, n)  \mathrm{d} s $$
    Figure~\ref{fig:dkb} shows the divergence as a function of sample sizes ranging from $n=2$ to $n=2^{14}$ and with $\lambda=1.1, \gamma=1,$ and $\rho=0.9$.
    We find that divergence is very insensitive to birth rates and sample proportion, so results are only shown for one scenario. 
    When $n=2$, the divergence is quite high, but it rapidly converges to zero. 
    We observe that to excellent approximation, the divergence scales in a very simple way as a function of sample size: $D(P_{ddcom}, P_{bdm} | \lambda, \mu, \rho, n) \approx e^{-3/2} / n$, and this is shown by the red line in Figure~\ref{fig:dkb}. 
    
    Figure~\ref{fig:dkb} also shows a comparison of the DDCoM marginal density of coalescent times with the BDM marginal likelihood with several different sample sizes and a smaller sample fraction of $\rho=0.01$. 
    When $n=2$, the distributions are quite different, but when $n=10$, they are very similar and when $n\geq 100$ they are almost indistinguishable.

    \begin{figure}
    \begin{center} \includegraphics[width=.9\textwidth]{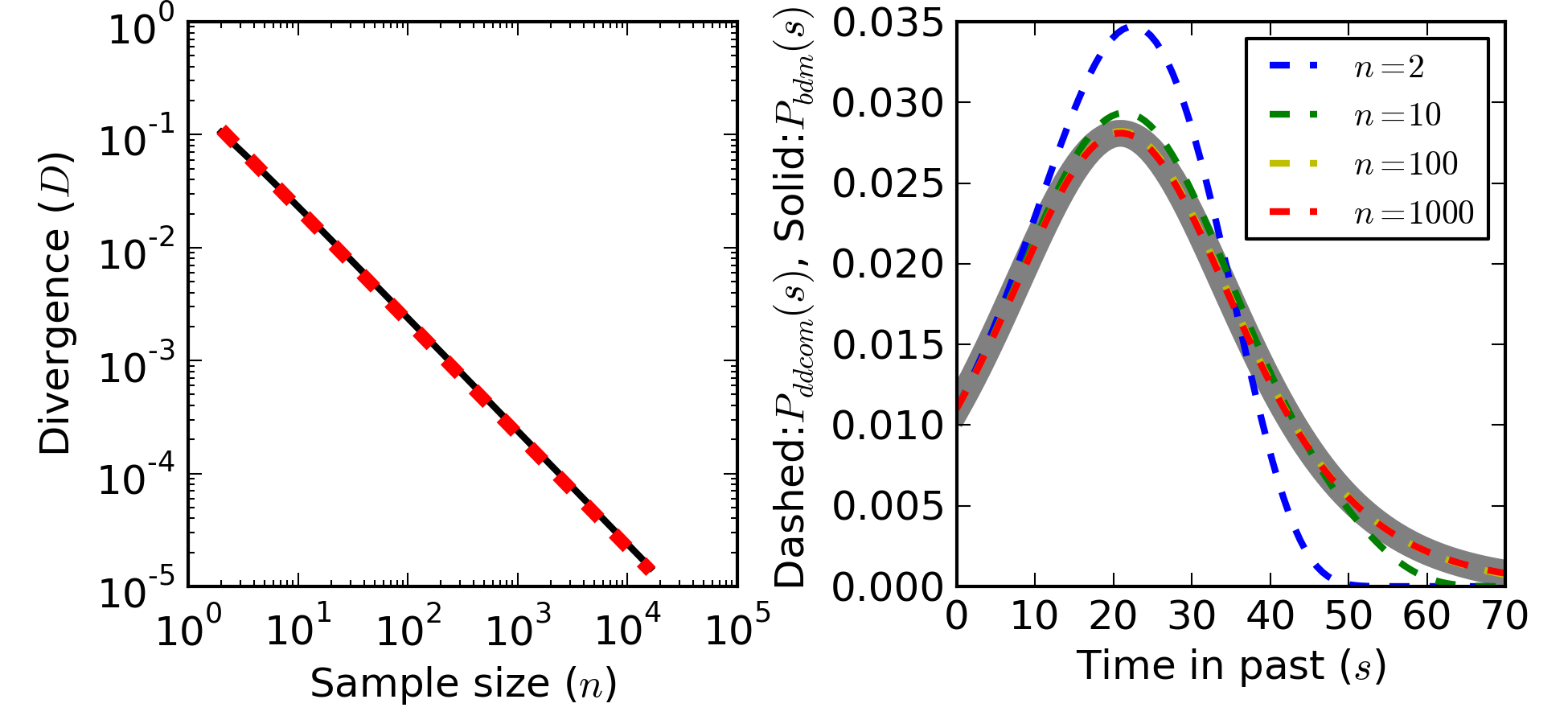} \end{center}
    \caption{ Left: The Kullback-Leibler divergence between the coalescent and birth-death distribution (black line) is shown versus the sample size on logarithmic axes. The red line shows a linear approximation ($e^{-3/2} / n$). $\lambda=1.1, \gamma=1$ and $\rho=.9$. Right: The birth-death marginal density of a node (grey) is compared with the coalescent density based on samples of size $n=2, 10, 100,$ and $1000$. $\lambda=1.1, \gamma=1$ and $\rho=0.01$.  \label{fig:dkb}} 
    \end{figure}

\section{Discussion}

Two distinct areas of concern have arisen related to phylodynamic inference using coalescent models (CoMs) and birth-death-sampling models (BDMs). 
CoMs based on a deterministic demographic process may be subject to inductive bias if the determinsitic process is a bad approximation to the true stochastic demographic process. 
Similarly, BDMs are subject to bias if the model of the sampling process is misspecified. 
We have found that the bias due to the deterministic approximation is generally very small for populations growing exponentially, 
 even when sampling 50\% of individuals from a small population. 
Furthermore, errors in CoMs due to a deterministic process can be resolved with additional computational effort, as it is possible to use the coalescent with a stochastic demographic process \cite{rasmussen2014phylodynamic,popinga2014bayesian}. 
Such methods may be necessary for populations with very small and noisy population dynamics. 
Bias is likely to be greatest if the population is small and growing slowly such that population dynamics are relatively noisy. 
Indeed, we found only one situation where the BDM was noticeably more precise than CoM estimators, which occurred with a small $R_0$ of 1.1 and a large sample fraction, 
however, we did not find a situation where the BDM estimator was substantially less biased than the CoM estimator.

We have found that BDMs can yield highly biased estimates if the sampling process is misspecified. 
It may be hard to detect if the sampling process deviates from the modeled form in many real world situations, and most real datasets are likely to violate the BDM sampling process assumptions to some degree. 
An example heterogeneous sampling through time is shown in Figure~\ref{fig:uksampprop} for a dataset which has previously been analyzed with BDMs in~\cite{kuhnert2014simultaneous}. 
Figure~\ref{fig:uksampprop} shows the sampling proportion through time of HIV sequence samples in the UK Resistance Database~\cite{brown2011transmission}
Typical for HIV sequence databases, the sample proportion is essentially zero throughout the 80s, and there is a rapid increase in sampling effort throughout the late 90s and early 00s, followed by a plateau after 2010  due to reporting delays.
In~\cite{kuhnert2014simultaneous} a BD SIR model was fitted to HIV sequence data from the UK under the assumption of a constant sampling rate, but the timespan of the estimated phylogenies ranged from 1978-2003 over which the true sampling rate varied greatly. 
Future work should explore how violation of sampling assumptions may bias estimates of $R_0$ when fitting BDM SIR models. 

\begin{figure} %
	\includegraphics[width=.5\textwidth]{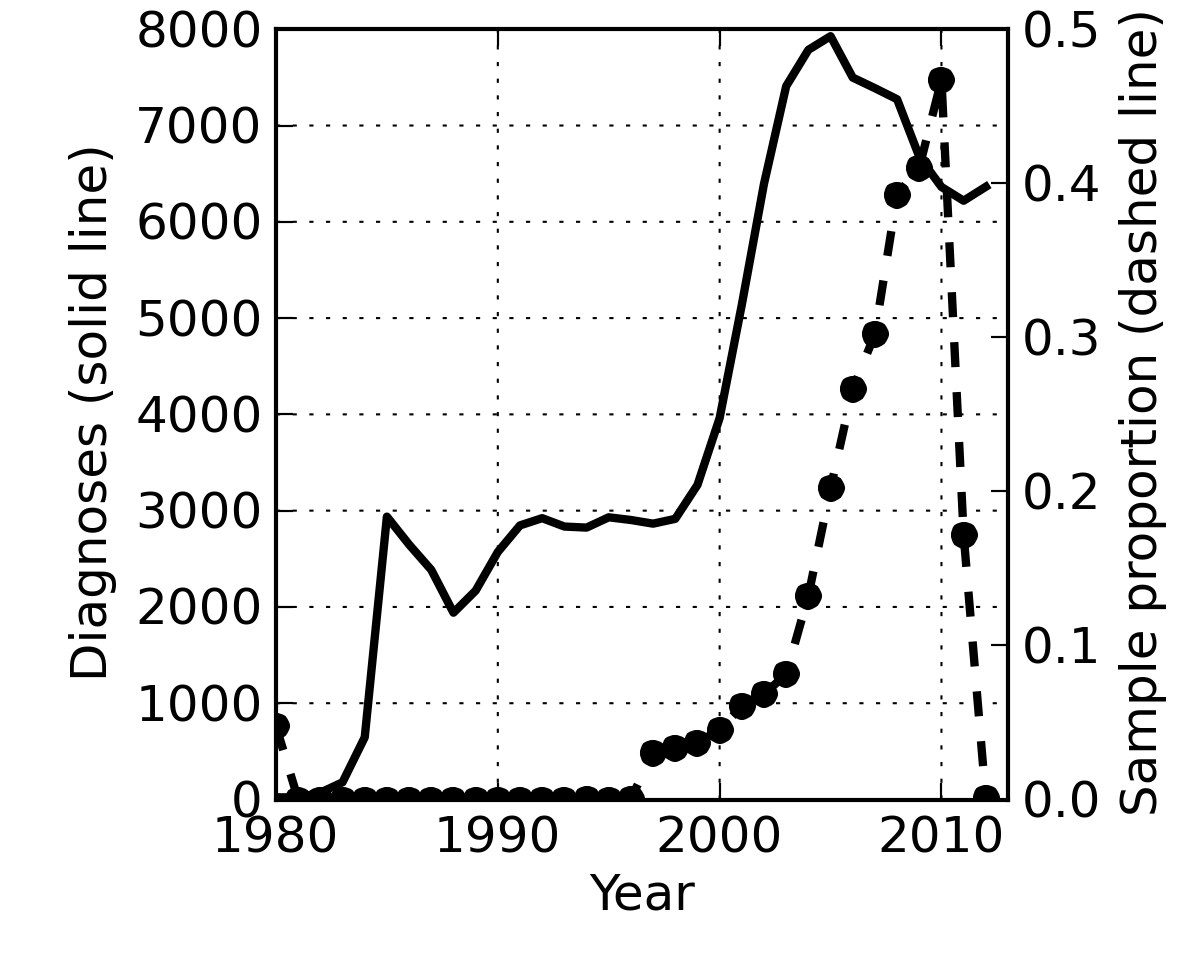}
	\caption{
	The HIV sequence sample rate through time using data from the HIV resistance database~\cite{brown2011transmission} and the number living diagnosed with HIV through time. 
	\label{fig:uksampprop}}
\end{figure}

The sequence of sample times may be informative about the population size through time if the sampling process can be correctly specified. 
We have shown how birth rates may be estimated from the sequence of sample times if sampling occurs according to the BDM assumptions, and this is possible even if the sample rate is not known. 
BDMs implicitly use the sequence of sample times to estimate birth and/or death rates, and this is the case even if the sampling rate is not given, but estimated. 
Comparisons of CoMs and BDMs should account for the effects of sampling, and a fair comparison can be obtained in the case of homochronous or serial-homochronous sampling with unknown sample rate, so that the the sample times contain no information about population size and birth rates.

Previous simulation-based studies on fitting susceptible-infected-recovered (SIR) epidemiological models to sequence data \cite{popinga2014bayesian} have purported to show increased statistical efficiency of BDMs relative to CoMs, 
but these studies did not control for the informativeness of sample times, and the supposed advantage of BDM in these simulations is likely due to the sampling model and not the genealogical model.
For example, the simulation studies in \cite{popinga2014bayesian} did not consider a homochronous sample, a misspecified sampling process, or the possibility of extending the coalescent estimators to use sample time information. 
The study in \cite{popinga2014bayesian} used a Bayesian method, in contrast to our ML methods, so some differences may also be due to the choice of priors. 
It was claimed in \cite{popinga2014bayesian} that the difference in performance of BDMs and CoMs was due to the latter's use of a misspecified deterministic demographic process, but in the context of exponential growth, we found very little bias due to the deterministic approximations of the coalescent, but large biases due to the effects of sampling. 

Future research on BDMs may reveal ways to accomodate more realistic sampling processes. 
For example, in \cite{stadler2013birth}, a piecewise constant sampling process was presented, however this also required the introduction of many more parameters to describe the sampling process. 
If the sampling process is known, a useful alternative to BDMs is model the sampling process in tandem with the coalescent.
As we have shown, the coalescent likelihood of a genealogy is approximately independent of the likelihood of the sample times, and for complex sampling processes it is much easier to model the genealogical and sampling process separately and combine likelihoods then to derive a joint likelihood. 
In most cases, the sampling process is unknown, but we have shown that CoMs are robust to a diverse range of sampling scenarios.

\textbf{Acknowledgements.} The authors thank Alison Brown (PHE UK) for providing HIV diagnosis statistics from the UK. Xavier Didelot and Caroline Colijn (Imperial College London) provided many helpful suggestions. SDWF is supported in part by an MRC Methodology Research Programme grant (MR/J013862/1). The funders had no role in study design, data collection and analysis, decision to publish, or preparation of the manuscript.

%~ bib
\bibliographystyle{vancouver}
\bibliography{volzETAL-bdCoSampling}

\appendix
\renewcommand{\thefigure}{S\arabic{figure}}
\renewcommand{\figurename}{Figure }
\setcounter{figure}{0}
		\begin{figure}
		{ \centering \includegraphics[width=.8\textwidth]{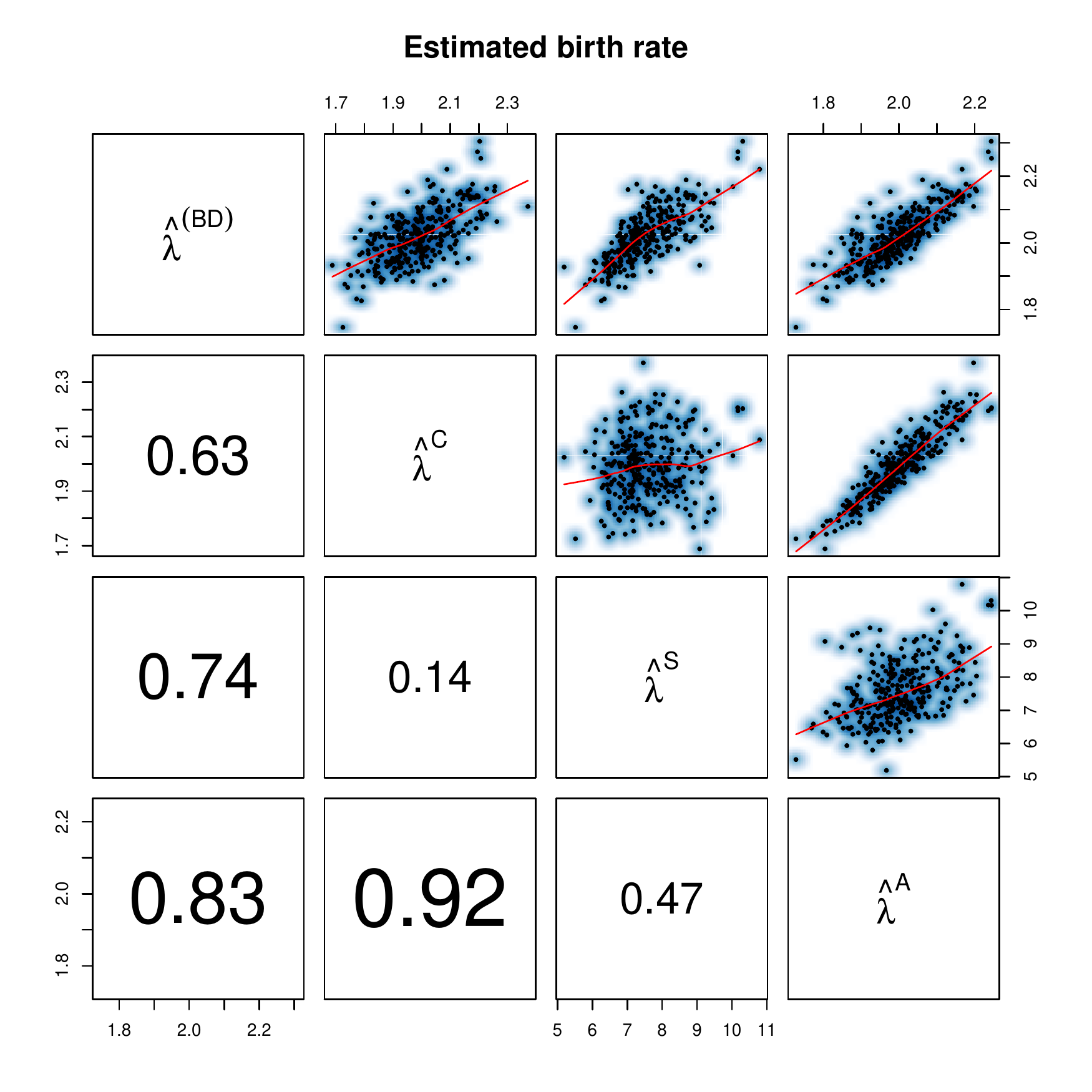}  }
		\caption{ The MLE birth rate using four estimation methods and based on 300 simulated genealogies. The Pearson correlation coefficient between the estimates is shown in the lower panels. The upper panels show a scatter plot with smoothing splines.   \label{fig:lambdapairs}}
		\end{figure}
		
		\begin{figure}
		{ \centering \includegraphics[width=.8\textwidth]{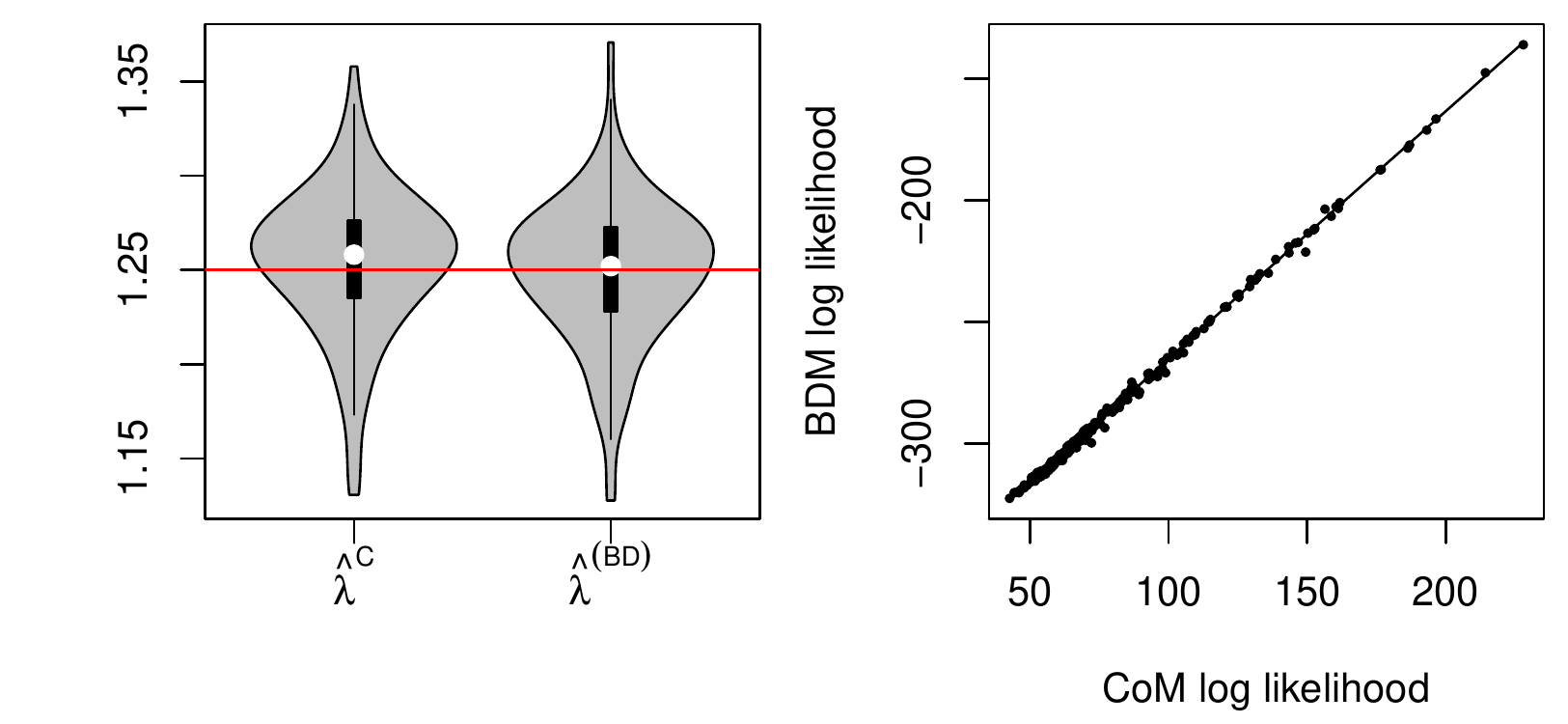} }
		\caption{ MLEs based on simulations with homochronous sampling and $\lambda=1.25$. Left: Distribution of MLE birth rates from 300 simulations and using BDM and CoM estimators. The red line shows the true parameter value.
		Right: The log likelihoods of the BDM and CoM MLEs with a smoothing spline. 
		 \label{fig:vioplot-homochronous1} }
		\end{figure}
		
		\begin{figure}
		\centering
		\includegraphics[width=.6\textwidth]{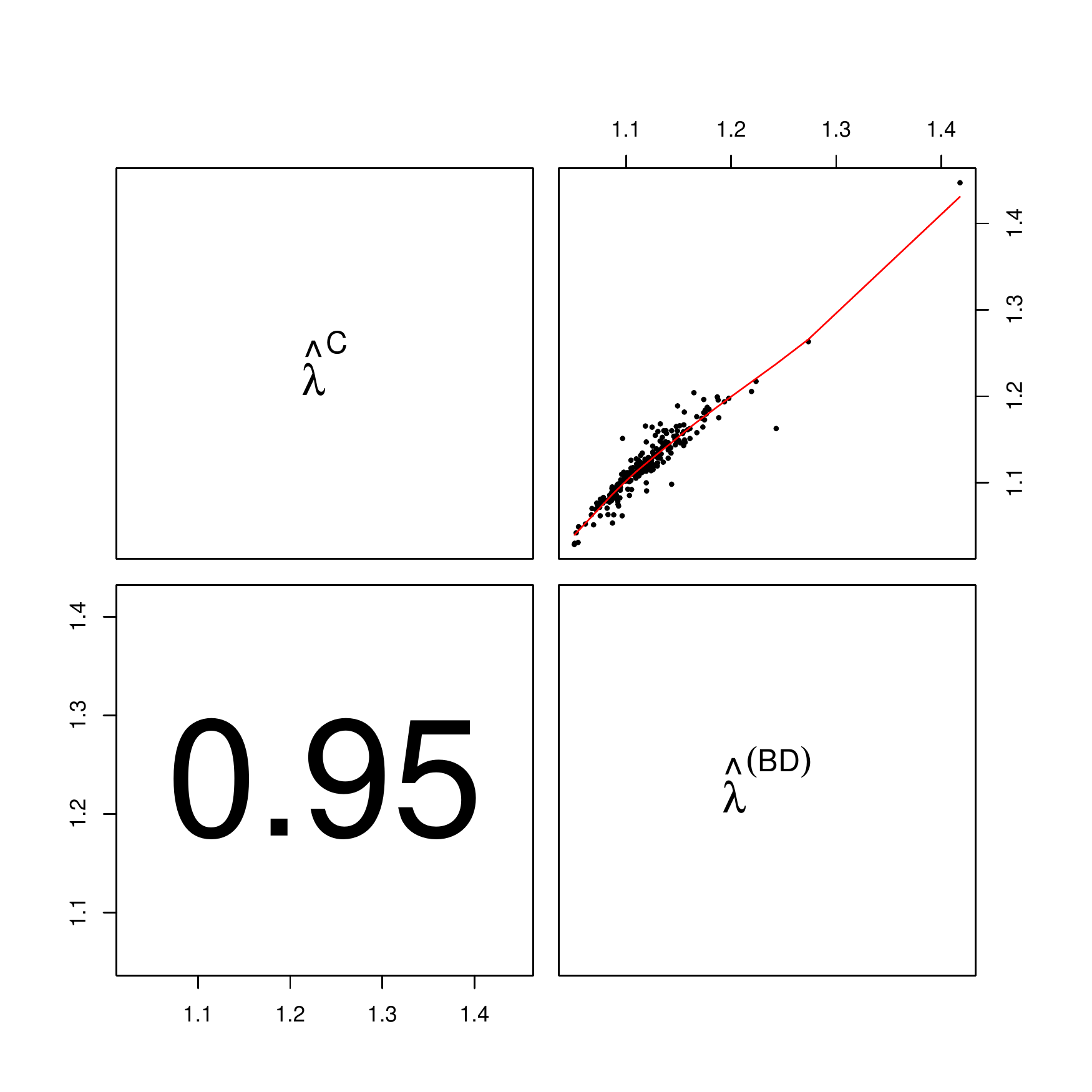}
		\caption{The BDM and CoM MLE birth rate based on 300 simulated genealogies with $\lambda=1.1$ and large sample fraction. The Pearson correlation coefficient between the estimates is shown in the lower panels. The upper panels show a scatter plot with smoothing splines.  }
		\label{fig:homochron2lambdaPairs}
		\end{figure}
\end{document}